%% file: main.tex
\documentclass[]{llncs}
\usepackage[bibliography=common]{apxproof}

\usepackage{amsfonts,amsmath,amssymb,amsthm}
\usepackage{paralist}
\usepackage{hyperref}
\usepackage{graphicx}
\usepackage{subfigure}
\usepackage{todonotes}
\usepackage{xspace}
\usepackage{complexity}
\usepackage{color, colortbl}
\usepackage{mathtools}
\usepackage{booktabs}
\usepackage{array}
\usepackage{bibentry}
\usepackage{thm-restate}
\usepackage{wrapfig}
\usepackage{enumitem}

\usepackage{thm-restate}

\usepackage[capitalise]{cleveref}
\Crefname{algorithm}{Algorithm}{Algorithms}
\Crefname{section}{Sect.}{Sects.}
\Crefname{observation}{Observation}{Observations}
\Crefname{redrule}{Reduction Rule}{Reduction Rules}
\Crefname{lemma}{Lemma}{Lemmas}
\Crefname{lemma2}{Lemma}{Lemmas}
\Crefname{theorem2}{Theorem}{Theorems}
\Crefname{claim}{Claim}{Claims}
\Crefname{cl}{Claim}{Claims}
\Crefname{claimx}{Claim}{Claims}
\Crefname{figure}{Fig.}{Figs.}
\Crefname{enumi}{Condition}{Conditions}
\Crefname{property}{Property}{Properties}
\Crefname{assumption}{Assumption}{Assumptions}

\DeclareMathOperator{\bw}{\rm bw}
\DeclareMathOperator{\tw}{\rm tw}
\DeclareMathOperator{\midset}{mid}

\newcommand{\Oh}{\mathcal{O}}

\newtheorem{redrule}{Reduction Rule}{}
\newtheorem{observation}{Observation}{}
\newtheorem{cl}{Claim}{}

\graphicspath{{./}}

\newcommand{\qedclaim}{\hfill $\blacksquare$}
\newenvironment{claimproof}{\noindent{\itshape Proof.}}{~\hfill \qedclaim \smallskip}
\definecolor{antiquewhite}{rgb}{0.98, 0.92, 0.84}

\newcommand{\LR}[1]{\left\{#1\right\}}

\newcommand{\gbar}{\overline{G}}

\begin{document}
\title{Parameterized and Approximation Algorithms for the Maximum Bimodal Subgraph Problem\thanks{Research started at the Dagstuhl Seminar 23162: New Frontiers of Parameterized Complexity in Graph Drawing, April 2023, and partially supported by: $(i)$ University of Perugia, Ricerca Base 2021, Proj. ``AIDMIX — Artificial Intelligence for
Decision Making: Methods for Interpretability and eXplainability''; $(ii)$ MUR PRIN Proj. 2022TS4Y3N - ``EXPAND: scalable algorithms for EXPloratory Analyses of heterogeneous and dynamic Networked Data'', $(iii)$ MUR PRIN Proj. 2022ME9Z78 - ``NextGRAAL: Next-generation algorithms for constrained GRAph visuALization'', $(iv)$ the Research Council of
Norway project BWCA 314528, $(v)$ the European Research Council (ERC) grant LOPPRE 819416, and $(vi)$ NSF-CCF 2212130}.
}

\titlerunning{Parameterized and Approximation Algorithms for the MBS Problem}


\author{
Walter Didimo\inst{1}\texorpdfstring{\href{https://orcid.org/0000-0002-4379-6059}{\protect\includegraphics[scale=0.45]{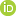}}}{}
\and
Fedor V. Fomin\inst{2}\texorpdfstring{\href{https://orcid.org/0000-0003-1955-4612}{\protect\includegraphics[scale=0.45]{orcid}}}{}
\and
Petr A. Golovach\inst{2}\texorpdfstring{\href{https://orcid.org/0000-0002-2619-2990}{\protect\includegraphics[scale=0.45]{orcid}}}{}
\and
Tanmay Inamdar\inst{2}\texorpdfstring{\href{https://orcid.org/0000-0002-0184-5932}{\protect\includegraphics[scale=0.45]{orcid}}}{}
\and
Stephen Kobourov\inst{3}\texorpdfstring{\href{https://orcid.org/0000-0002-0477-2724}{\protect\includegraphics[scale=0.45]{orcid}}}{}
\and
Marie Diana Sieper\inst{4}\texorpdfstring{\href{https://orcid.org/0009-0003-7491-2811}{\protect\includegraphics[scale=0.45]{orcid}}}{}
}

\date{}

\institute{
Dept.~of Engineering, University of Perugia, Italy\\
\email{walter.didimo@unipg.it}\and
Department of Computer Science, University of Bergen, Norway\\
\email{fedor.fomin,petr.golovach,tanmay.inamdar@uib.no}\and
Department of Computer Science, University of Arizona, USA\\
\email{kobourov@cs.arizona.edu}
\and
Department of Computer Science, University of W\"urzburg, Germany\\
\email{marie.sieper@uni-wuerzburg.de}}
\maketitle

\begin{abstract}

%
%

A vertex of a plane digraph is \emph{bimodal} if all its incoming edges (and hence all its outgoing edges) are consecutive in the cyclic order around it. 
A plane digraph is bimodal if all its vertices are bimodal.
Bimodality is at the heart of many types of graph layouts, such as upward drawings, level-planar drawings, and L-drawings. If the graph is not bimodal, the \emph{Maximum Bimodal Subgraph (MBS)} problem asks for an embedding-preserving bimodal subgraph with the maximum number of edges.
We initiate the study of the MBS problem from the parameterized complexity perspective with two main results: (i) we describe an FPT algorithm parameterized by the branchwidth (and hence by the treewidth) of the graph; (ii) we establish that MBS parameterized by the number of non-bimodal vertices admits a polynomial kernel. 
As the byproduct of these results, we obtain a subexponential FPT algorithm and an efficient polynomial-time approximation scheme for MBS.

\keywords{bimodal graphs, maximum bimodal subgraph, parameterized complexity, FPT algorithms, polynomial kernel, approximation scheme}
\end{abstract}

\section{Introduction}\label{se:introduction}
Let $G$ be a plane digraph, that is, a planar directed graph with a given planar embedding. A vertex $v$ of $G$ is \emph{bimodal} if all its incoming edges (and hence all its outgoing edges) are consecutive in the cyclic order around $v$. In other words, $v$ is bimodal if the circular list of edges incident at $v$ can be split into at most two linear lists, where all edges in the same list are either all incoming or all outgoing~$v$. 
Graph $G$ is \emph{bimodal} if all its vertices are bimodal. Bimodality is a key property at heart of many graph drawing styles. In particular, it is a necessary condition for the existence of \emph{level-planar} and, more generally, \emph{upward planar} drawings, where the edges are represented as curves monotonically increasing in the upward direction according to their orientations~\cite{DBLP:books/ph/BattistaETT99,DBLP:journals/tsmc/BattistaN88,DBLP:reference/algo/Didimo16,DBLP:conf/gd/JungerLM98}; see \Cref{fi:intro-a}. 
 Bimodality is also a sufficient condition for \emph{quasi-upward planar} drawings, in which edges are allowed to violate the upward monotonicity a finite number of times at points called \emph{bends}~\cite{DBLP:journals/algorithmica/BertolazziBD02,DBLP:conf/gd/BinucciGLT21,DBLP:journals/cj/BinucciD16}; see \Cref{fi:intro-b}. 
 It has been shown that bimodality is also a sufficient condition for the existence of \emph{planar L-drawings} of digraphs, in which distinct L-shaped
 edges may overlap but not cross~\cite{DBLP:journals/jgaa/AngeliniCCL22,DBLP:conf/mfcs/AngeliniCCL22,DBLP:journals/ijfcs/KariOALBDPRT18a}; see \Cref{fi:intro-c}. 
A generalization of bimodality is $k$-modality. Given a positive even integer $k$, a plane digraph is \emph{$k$-modal} if the edges at each vertex can be grouped into at most $k$ sets of consecutive edges with the same orientation~\cite{DBLP:conf/esa/VialLG19}. 
  In particular, it is known that $4$-modality is necessary for planar L-drawings~\cite{DBLP:conf/gd/ChaplickCCLNPTW17}.       

While testing if a digraph $G$ admits a bimodal planar embedding can be done in linear time~\cite{DBLP:journals/algorithmica/BertolazziBD02}, a natural problem that arises when $G$ does not have such an embedding is to extract from $G$ a subgraph of maximum size (i.e., with the maximum number of edges) that fulfills this property. This problem 
is NP-hard, even if $G$ has a given planar embedding and we look for an embedding-preserving maximum bimodal subgraph~\cite{DBLP:journals/comgeo/BinucciDG08}. We address exactly this fixed-embedding version of the problem, and call it the \emph{Maximum Bimodal Subgraph} (\textsc{MBS}) problem.

\begin{figure}[tb]
    \centering
    \subfigure[]{\includegraphics[page=1,width=.32\textwidth]{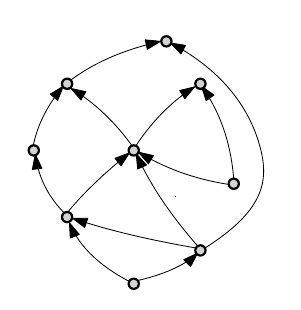}\label{fi:intro-a}}
    \hfil
    \subfigure[]{\includegraphics[page=2,width=.32\textwidth]{intro}\label{fi:intro-b}}
    \hfil
    \subfigure[]{\includegraphics[page=3,width=.32\textwidth]{intro}\label{fi:intro-c}}
    \caption{(a) An upward planar drawing. (b) A quasi-upward planar drawing, where edge $e$ makes two bends (the two horizontal tangent points). (c) A bimodal digraph (above) and a corresponding planar L-drawing (below).}
    \label{fi:intro}
\end{figure}

\smallskip\noindent{\bf Contribution.} While a heuristic and a branch-and-bound algorithm are given in~\cite{DBLP:journals/comgeo/BinucciDG08} to solve MBS (and also to find a maximum upward-planar digraph), here we study this problem from the parameterized complexity and approximability perspectives (refer to~\cite{CyganFKLMPPS15,FominLSZ19} for an introduction to  parameterized complexity). More precisely, we consider the following more general version of the problem with weighted edges; it coincides with MBS when we restrict to unit edge weights.  

\smallskip \noindent\textsc{MWBS$(G,w)$} (\emph{Maximum Weighted Bimodal Subgraph}). \emph{Given a plane digraph~$G$ and an edge-weight function $w: E(G) \rightarrow \mathbb{Q}^+$, compute a bimodal subgraph of $G$ of maximum weight, i.e., whose sum of the edge weights is maximum over all bimodal subgraphs of $G$.} 


\smallskip\noindent Our contribution can be summarized as follows.

\smallskip \noindent\textsl{$-$ Structural parameterization.} 
We show that  \textsc{MWBS} is FPT  when parameterized by the \emph{branchwidth} 
of the input digraph $G$ or, equivalently, by the \emph{treewidth} of $G$ (\Cref{se:fpt-branchwidth}). Our algorithm deviates from a standard dynamic approach for graphs of bounded treewidth. The main difficulty here is that we have to incorporate the ``topological'' information about the given embedding
in the dynamic program. 
We accomplish this via the sphere-cut decomposition of Dorn et al.~\cite{DBLP:journals/algorithmica/DornPBF10}. 
  
\smallskip \noindent\textsl{$-$ Kernelization.} Let $b$ be the number of non-bimodal vertices in an input digraph~$G$.  We construct a polynomial kernel for the decision version of \textsc{MWBS} parameterized by $b$ (\cref{se:fpt-b}). Our kernelization algorithm performs in several steps. First we show how to reduce the instance to an equivalent instance whose branchwidth is $\Oh(\sqrt{b})$. Second, by using specific gadgets, we compress the problem to an instance of another problem whose size is bounded by a polynomial of $b$. In other words, we provide a polynomial compression for \textsc{MWBS}. Finally, by the standard arguments, \cite[Theorem~1.6]{FominLSZ19},  based on a polynomial reduction between any NP-complete problems, we obtain a polynomial kernel for \textsc{MWBS}. 
 
 
\smallskip
By pipelining the crucial step of the  kernelization algorithm with the branchwidth algorithm,  
we obtain a parameterized subexponential algorithm for \textsc{MWBS} of running time $2^{\Oh(\sqrt{b})}\cdot n^{\Oh(1)}$. Since 
$b\leq n$, this also implies an algorithm of running time  $2^{\Oh(\sqrt{n})}$.
%
%
Note that our algorithms are asymptotically optimal up to the \emph{Exponential Time Hypothesis} (ETH)~\cite{ImpagliazzoP99,ImpagliazzoPZ01}.  
The NP-hardness result of MBS (and hence of MWBS) given in~\cite{DBLP:journals/comgeo/BinucciDG08} exploits a reduction from \textsc{Planar-3SAT}. The number of non-bimodal vertices in the resulting instance of MBS is linear in the size of the \textsc{Planar-3SAT} instance. Using the standard techniques for computational lower bounds for problems on planar graphs \cite{CyganFKLMPPS15}, we obtain that the existence of an $2^{o(\sqrt{b})} \cdot n^{\Oh(1)}$-time algorithm for MBWS would contradict ETH.

\smallskip\noindent\textsl{$-$ Approximability.}
We provide an Efficient Polynomial-Time Approximation Scheme (EPTAS) for \textsc{MWBS}, based on Baker's (or shifting) technique~\cite{Baker94}. Namely, using our algorithm for graphs of bounded branchwidth, we give an $(1+\epsilon)$-approximation algorithm that runs in $2^{\Oh(1/\epsilon)} \cdot n^{\Oh(1)}$ time. 

\smallskip
Full proofs of the results marked with an asterisk (*), as well as additional definitions and technical details, are given in appendix.
\section{ Definitions and Terminology}\label{se:basic}
Let $G$ be a digraph. We denote by $V(G)$ and $E(G)$ the set of vertices and the set of edges of $G$.
Throughout the paper we assume that $G$ is planar and that it comes with a planar embedding; such an embedding fixes, for each vertex $v \in V(G)$, the clockwise order of the edges incident to $v$.
We say that $G$ is a \emph{planar embedded digraph} or simply that $G$ is a \emph{plane digraph}. 

\medskip\noindent{\bf Branch decomposition and sphere-cut decomposition.}
A \emph{branch decomposition} of a graph $G$ defines a hierarchical clustering of the edges of $G$, represented by an unrooted proper binary tree, that is a tree with non-leaf nodes of degree three, whose leaves are in one-to-one correspondence with the edges of $G$. More precisely, a branch decomposition of $G$ consists of a pair $\langle T,\xi \rangle$, where $T$ is an unrooted proper binary tree and $\xi : \mathcal L(T) \leftrightarrow E(G)$ is a bijection between the set $\mathcal L(T)$ of the leaves of~$T$ and the set $E(G)$ of the edges of $G$. 
For each arc $a$ of $T$, denote by $T^a_1$ and $T^a_2$ the two connected components of $T \setminus \{a\}$, and, for $i=1,2$, let $G^a_i$ be the subgraph of $G$ that consists of the edges corresponding to the leaves of $T^a_i$. 
The \emph{middle set} $\midset(a) \subseteq V (G)$ is 
the intersection of the vertex sets of $G^a_1$ and $G^a_2$, i.e., $\midset(a) := V(G^a_1) \cap V(G^a_2)$.
The \emph{width} $\beta(\langle T,\xi \rangle)$ of $\langle T,\xi \rangle$ is the maximum size of the middle sets over all arcs of $T$, i.e., $\beta(\langle T,\xi \rangle) = \max\{\lvert \midset(a) \rvert: a \in E(T)\}$. An \emph{optimal branch decomposition} of $G$ is a branch decomposition with minimum width; this width is called the \emph{branchwidth} of~$G$ and is denoted by~$\bw(G)$.

\begin{figure}[tb]
    \centering
    \subfigure[]{\includegraphics[page=1,width=.35\textwidth]{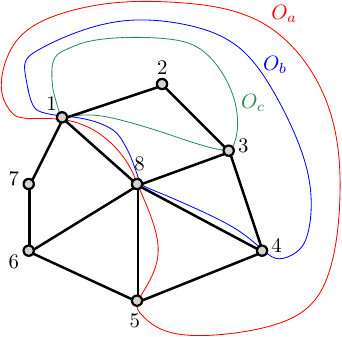}\label{fi:scd-a}}
    \hfil
    \subfigure[]{\includegraphics[page=2,width=.35\textwidth]{scd}\label{fi:scd-b}}
    \caption{A plane graph $G$ and a sphere-cut decomposition of $G$; three nooses are highlighted on $G$ for the arcs $a$, $b$, and $c$ of the decomposition tree.}
    \label{fi:scd}
\end{figure}

\smallskip
A sphere-cut decomposition is a special type of branch decomposition (see \Cref{fi:scd}). Let $G$ be a connected planar graph, topologically drawn on a sphere $\Sigma$. A \emph{noose} $O$ of $G$ is a closed simple curve on $\Sigma$ that intersects $G$ only at vertices and that traverses each face of $G$ at most once. The \emph{length} of $O$ is the number of vertices that $O$ intersects. Note that, $O$ bounds two closed discs $\Delta_O^1$ and $\Delta_O^2$ in $\Sigma$; we have $\Delta_O^1 \cap \Delta_O^2 = O$ and $\Delta_O^1 \cup \Delta_O^2 = \Sigma$. Let $\langle T,\xi \rangle$ be a branch decomposition of $G$. Suppose that for each arc $a$ of $T$ there exists a noose $O_a$ that traverses exactly the vertices of $\midset(a)$ and whose closed discs $\Delta_{O_a}^1$ and $\Delta_{O_a}^2$ enclose the drawings of $G_1^a$ and of $G_2^a$, respectively. Denote by $\pi_a$ the circular clockwise order of the vertices in $\midset(a)$ along $O_a$ and let $\Pi=\{\pi_a : a \in E(T)\}$ the set of all circular orders $\pi_a$. The triple $\langle T, \xi, \Pi \rangle$ is a \emph{sphere-cut decomposition} of $G$. We 
assume that the vertices of $\midset(a)=V(G_1^a) \cap V(G_2^a)$ are enumerated according to $\pi_a$. 
Since a noose $O_a$ traverses each face of $G$ at most once, both graphs $G_1^a$ and $G_2^a$ are connected. Also, the nooses are pairwise non-crossing, i.e., for any pair of nooses $O_a$ and $O_b$, we have that $O_b$ lies entirely inside $\Delta_{O_a}^1$ or entirely inside $\Delta_{O_a}^2$. 
For a noose $O_a$, we define $\midset(O_a)=\midset(a)$, or in general, we define $\midset(\phi)$ to be the vertices cut by $\phi$.
We rely on the following result on the existence and computation of a sphere-cut decomposition~\cite{DBLP:conf/wg/JacobP22} (see also  \cite{DBLP:journals/algorithmica/DornPBF10}).

\begin{proposition}[\cite{DBLP:conf/wg/JacobP22}] \label{th:sphere-cut-with-deg-1}
    Let $G$ be a connected graph embedded in the sphere with $n$ vertices and  branchwidth $\ell \geq 2$. Then there exists a sphere-cut decomposition of $G$ with width $\ell$, and it can be computed in $\Oh(n^3)$ time.
\end{proposition}

We remark that the branchwidth $\bw(G)$ and the treewidth $\tw(G)$ of a graph $G$ are within a constant factor: $ \bw(G) - 1 \leq \tw(G) \leq \lfloor \frac{3}{2}\bw(G) \rfloor - 1$ (see  \cite{DBLP:journals/jct/RobertsonS91}).

\section{FPT Algorithms for MWBS by Branchwidth}\label{se:fpt-branchwidth}

In this section we describe an FPT algorithm parameterized by branchwidth. 
We first introduce configurations, which encode on which side of a closed curve and in what order in a bimodal subgraph for a vertex $v$ the switches between incoming to outgoing edges happen.

\begin{definition}[Configuration]
\label{def:configuration}
    Let $C=\{(i), (o), (i,o), (o,i), (o,i,o), (i,o,i)\}$.
    Let $G$ be a graph embedded in the sphere $\Sigma$, $\phi$ be a noose in $\Sigma$ with a prescribed inside, $v \in \midset{(\phi)}$, and $X \in C$.
    Let $E^{v,\phi}$ be the set of edges incident to $v$ in $\phi$. We say $v$ has \emph{configuration} $X$ in $\phi$, if $E^{v,\phi}$ can be partitioned into sets such~that:
    \begin{enumerate}
        \item For every $x\in X$, there is a (possibly empty) set $E_x$ associated with it.
        \item Every set associated with an $i$ ($o$) contains only in- (/out-) edges of $v$.
        \item For every set, the edges contained in it are successive around $v$.
        \item The sets $E_x$ appear clockwise (seen from $v$) in the same order in $G$ inside $\phi$ as the $x$ appear in $X$.
    \end{enumerate}
    For every $v \in \midset{(\phi)}$, let $X_v$ be a configuration of $v$ in $\phi$. We say $X_{\phi} = \{X_v \mid v\in \midset{(\phi)}\}$ is a \emph{configuration set} of $\phi$.
\end{definition}
If $G$ is bimodal, then for every noose $\phi$ and every vertex $v\in \midset{(\phi)}$, $v$ must have at least one configuration $X \in C$ in $\phi$.
Note that configurations and configuration sets are not unique, as seen in \Cref{fi:config}. A vertex can even have all configurations if it has no incident edges in $\phi$.
%
%
The next definition is needed to encode when configurations can be combined in order to obtain bimodal vertices.

\begin{definition}[Compatible configurations]
    Let $X, X', X^* \in C$ be configurations.
    We say $X, X'$ are \emph{compatible configurations} or short \emph{compatible}, if by concatenating $X, X'$ and deleting consecutive equal letters, the result is a substring of $(o,i,o)$ or $(i,o,i)$. Note that it is not important in which order we concatenate $X, X'$.
    See Figure \ref{fi:compatible1}.
    %
    We say $X$ and $X'$ are \emph{compatible with respect to $X^*$} if by concatenating $X, X'$ (in this order) and deleting consecutive equal letters, the result is a substring of $X^*$.
\end{definition}

A configuration $X$ can have several compatible configurations, for example $(i, o) \in C$ is compatible with $(o), (i)$ and $(o,i)$. From these $(o, i)$ is in some sense maximal, meaning that configurations $(o)$ and $(i)$ are substrings of $(o, i)$. Given a configuration $X$, a \emph{maximal compatible configuration} $X'$ of $X$ is a configuration that is compatible with $X$, and all other compatible configurations of $X$ are substrings of $X'$.
Observe that every configuration has a unique maximal compatible configuration, they are pairwise:
$(i)-(i, o, i)$,  $(o)-(o,i,o)$ and $(o,i)-(i,o)$.

We say a noose $\phi_3$ is \emph{composed} of the nooses $\phi_1$ and $\phi_2$, if the edges of $G$ in $\phi_3$ are partitioned by $\phi_1$ and $\phi_2$.
If a noose $\phi_3$ is composed of nooses $\phi_1$ and $\phi_2$, and there exists a vertex $v \in \midset(\phi_1) \cap \midset(\phi_2) \cap \midset(\phi_3)$, such that in $\phi_3$ around $v$, all adjacent edges of $v$ in $\phi_1$ are clockwise before all adjacent edges of $v$ in $\phi_2$.
If $X, X'$ and $X^*$ are nooses and $X$ and $X'$ are compatible  with respect to $X^*$, and $v$ has configuration $X$ in $\phi_1$ and configuration $X'$ in $\phi_2$, then it has configuration $X^*$ in $\phi_3$.
See Figure \ref{fi:compatible2}.


%
If a curve $\phi$ contains only one edge on its inside, finding maximal subgraphs for a configuration inside $\phi$ is easy.
\begin{figure}[tb]
    \centering
    \subfigure[]{\includegraphics[height=.33\textwidth]{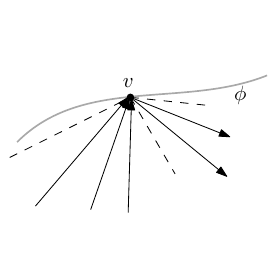}\label{fi:config}}
    \hfil
    \subfigure[]{\includegraphics[height=.33\textwidth]{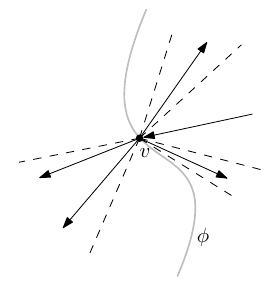}\label{fi:compatible1}}
    \hfil
    \subfigure[]{\includegraphics[height=.33\textwidth]{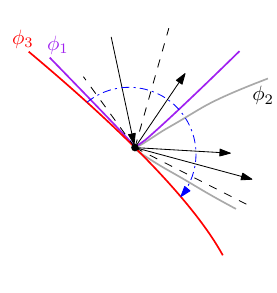}\label{fi:compatible2}}
    \caption{
    (a) A vertex with configurations $(o,i), (o,i,o)$ and $ (i,o,i)$ in $\phi$. The most restricted and thus minimal configuration is $(o,i)$.
    (b) A vertex with configuration $(o,i,o)$ in $\phi$ and $(o)$ outside of $\phi$.
    Concatenating $(o,i,o)$ with $(o)$ and deleting consecutive equal letters results in $(o,i,o)$, the result is a substring of $(o,i,o)$, thus $(o,i,o)$ and $(o)$ are compatible.
    (c) Note that $\phi_3$ is composed of $\phi_1$ and $\phi_2$;
    the inside of $\phi_1$, the inside of $\phi_2$ and the outside of $\phi_3$ are clockwise in this order around $v$ with configuration $(i, o)$ in $\phi_1$ and $(o)$ in $\phi_2$. They can be concatenated to configuration $(i, o)$ in $\phi_3$, while $(i, o)$ and $(o)$ are compatible w.r.t. $(i,o)$, but not~$(o,i)$.}
    \label{fi:compatible}
\end{figure}
\begin{restatable}[*]{lemma}{lemconfigsolutionbase}
    \label{lem:config_solution_base}
    Let $G$ be a graph embedded in the sphere $\Sigma$, let $e = \{u, v\}$ be an edge and let $\phi$ be a noose that cuts $G$ only in $u$ and $v$, such that $e$ is in $\phi$ and all other edges are on the outside of $\phi$.
    Let $X_u, X_v$ be prescribed configurations. Then we can compute in $\Oh(1)$ time the maximum subgraph $G'$ of $G$ such that $u, v$ have configuration $X_u$ respectively $X_v$ in $\phi$ in $G'$.
\end{restatable}

We will now see how we can compute optimal subgraphs bottom-up.
\begin{restatable}[*]{lemma}{lemconfigsolutiondpstep}
    \label{lem:config_solution_dpstep}
    Let $G$ be a graph embedded in the sphere $\Sigma$, let $\phi_1, \phi_2, \phi_3$ be nooses with length at most $\ell$ each, and let $E_{\phi_1}, E_{\phi_2}, E_{\phi_3}$ be the sets of edges contained inside the respective noose with $E_{\phi_1}, E_{\phi_2}$ being a partition of $E_{\phi_3}$. Let $X_{\phi_3}$ be a configuration set for $\phi_3$.
    Let further for every configuration set $X_{\phi_1}$ ($ X_{\phi_2}$) of $\phi_1$ ($\phi_2$), the maximum subgraph that has configuration set $X_{\phi_1}$ ($X_{\phi_2}$) and is bimodal in $\phi_1$ ($\phi_2$) be known.
    Then a maximum subgraph $G'$ of $G$ that has configuration set $X_{\phi_3}$ and is bimodal in $\phi_3$ can be computed in $\Oh(6^{2\ell}) \cdot n^{\Oh(1)}$~time.
\end{restatable}

If a noose $\phi$ contains only $e\in E$, we have only two options in $\phi$: delete $e$ or do not. Testing which is optimal can be done in constant time, this leads to \Cref{lem:config_solution_base}.
Now let $\phi_3$ be a noose that contains more than one edge, let $\phi_1, \phi_2$ be two nooses that partition the inside of $\phi_3$, and let $X_{\phi_3}$ be a given configuration set. If we already know optimal solutions for any given configuration set in $\phi_1$ ($\phi_2$) (which we already computed when traversing the sphere-cut decomposition bottom up), we can guess for some optimal solution for $\phi_3$ for every $v \in \midset(\phi_1) \cap \midset(\phi_2)$ the configuration it has in $\phi_1$ and in $\phi_2$. This gives us configuration sets $X_{\phi_1}$ and $X_{\phi_2}$ for $\phi_1$ and $\phi_2$, respectively (for every $v\in \midset(\phi_1)\setminus \midset(\phi_2)$ we take its configuration in $X_{\phi_3}$). We obtain the corresponding solution $G'$ that coincides with the optimal solution for $\phi_1$ ($\phi_2$) in $\phi_1$ ($\phi_2$) respecting $X_{\phi_1}$ ($X_{\phi_2}$) and that coincides with $G$ outside of $\phi_3$.
Since $|\midset(\phi_1) \cap \midset(\phi_2)| \leq \ell$, we achieve the same by enumerating all possible configurations for $\midset(\phi_1) \cap \midset(\phi_2)$, compute the corresponding solutions and take the maximum in $\Oh(6^{2\ell})\cdot n^{\Oh(1)}$ time, leading to \Cref{lem:config_solution_dpstep}.
%
We now obtain the following theorem.

\begin{restatable}[*]{theorem}{thmmwbsfptbw} \label{thm:mwbs-fpt-bw}
    There is an algorithm that solves \textsc{MWBS$(G, w)$} in $2^{\Oh(\bw(G))} \cdot n^{\Oh(1)}$ time. In particular, \textsc{MWBS} is FPT when parameterized by~branchwidth.
\end{restatable}
\begin{proof}[Sketch]
Assume that $G$ is connected (otherwise process every connected component independently). If $\bw(G) = 1$, $G$ is a star and we can compute an optimal solution in polynomial time. Otherwise, according to \Cref{th:sphere-cut-with-deg-1} we can compute a sphere-cut decomposition $\langle T, \xi, \Pi \rangle$ for $G$ with optimal width $\ell$. We  pick any leaf of $T$ to be the root $r$ of $T$. For every noose $O$ corresponding to an arc of $T$ let $X_O$ be a configuration set for $O$. Then we define $E_{(O, X_O)}$ to be edge set of minimum weight, such that $G \setminus E_{(O, X_O)}$ is bimodal inside of $O$ and has configuration set $X_O$ in $O$.
We now compute the $E_{(O, X_O)}$ bottom-up.
For a noose $O$ corresponding to a leaf-arc in $T$, \Cref{lem:config_solution_base} shows that we can compute all possible values of $E_{(O, X_O)}$ in linear time.
For a noose $O$ corresponding to a non-leaf arc in $T$, \Cref{lem:config_solution_dpstep} shows that we can compute $E_{O, X_O}$ for a given $X_O$ in $\Oh(6^{2\ell}) \cdot n^{\Oh(1)}$ time, and thus all entries for $O$ in $\Oh(6^{3\ell}) \cdot n^{\Oh(1)}$ time.
Let $e\in E$ be the edge associated with $r$. We have only two options left, delete $e$ or do not. In both cases we obtain the optimal solution for the rest of $G$ from the values $E_{(O, X_O)}$.
The overall running time is $2^{\Oh(\ell)} \cdot n^{\Oh(1)}$.
\end{proof}


Since our input graphs are planar, we immediately obtain a subexponential algorithm for MWBS because 
for a planar graph $G$, $\bw(G) = \Oh(\sqrt{n})$\cite{FominT06}.

\begin{theorem} \label{thm:mwbs--bw}
    \textsc{MWBS$(G=(V,E), w)$} can be solved in $2^{\Oh(\sqrt{n})}$ time. 
\end{theorem}

\section{Compression for MWBS by $b$}\label{se:fpt-b}
Throughout this section we assume that (i) the weights are rational, that is, for $(G,w)$, $w\colon V(G)\rightarrow \mathbb{Q}^+$ and (ii) we consider the decision version of MWBS, that is, additionally to $(G,w)$, we are given a target value $W\in\mathbb{Q}^+$ and the task is to decide whether $G$ has a bimodal subgraph $G^*$ with $w(E(G^*))\geq W$.

\medskip\noindent{\bf Further definitions.} 
For simplicity, we say that a bimodal vertex of $G$ is a \emph{good} vertex, and that a non-bimodal vertex is a \emph{bad} vertex. We denote by  
${\cal G}(G)$ and ${\cal B}(G)$ the sets of good and bad vertices of $G$, respectively.
Given a vertex $v \in V(G)$, an \emph{in-wedge} (resp. \emph{out-wedge}) of $v$ is a maximal circular sequence of consecutive incoming (resp. outgoing) edges of $v$. 
Clearly, if $v$ is bimodal it has at most one in-wedge and at most one out-wedge.   
%
	Given a vertex $v \in {\cal B}(G)$, a \emph{good edge-section} of $v$ is a maximal consecutive sequence of in- and out- wedges of $v$, such that no edge is incident to another bad vertex.

\begin{observation}\label{obs:bounded-sections}
	Let $(G, w)$ be an instance of MWBS with $b$ bad vertices, and let $v \in {\cal B}(G)$. Then $v$ can have at most $b-1$ good edge-sections.
\end{observation}

We introduce a generalization of MWBS called
\textsc{Cut-MWBS$(G,w,\mathcal{E})$} (\emph{maximum weighted bimodal subgraph with prescribed cuts}). Given a plane digraph $G$, an edge-weight function $w: E(G) \rightarrow \mathbb{Q}^+$, and  a partition $\mathcal{E}$ of $E(G)$, compute a bimodal subgraph $G'$ of $G$ of maximum weight, i.e., whose sum of the edge weights is maximum over all bimodal subgraphs of $G$, under the condition that for every set $E_i \in \mathcal{E}$, either all $e \in E_i$ are still present in $G'$ or none of them are.
We can see that every instance $(G,w)$ of MWBS is equivalent to the instance $(G,w,\{\{e\}\mid e \in E(G)\})$ of Cut-MWBS, and thus Cut-MWBS is NP-hard.  Also, the decision variant of the problem is NP-complete. 


We now give reduction rules for the MWBS to Cut-MWBS compression, and prove that each of them is \emph{sound}, i.e., it can be performed in polynomial time and the reduced instance is solvable if and only if the starting instance is~solvable. 

\begin{redrule} \label{redrule:isolated-comp}
    Let $(G, w)$ be an instance of MWBS,  and $v \in V(G)$ be an isolated vertex. Then, let $(G', w)$ be the new instance, where $V(G') = V(G) \setminus \{v\}$.
\end{redrule}
\begin{redrule} \label{redrule:bimodal-comp}
    Let $(G, w)$ be an instance of MWBS with the target value $W$, and $u, v \in {\cal G}(G)$ be such that $(u, v)$ is an edge. Then, the resulting instance is $(G', w)$, where $G' = G - (u, v)$, and the new target value is $W'=W-w(u,v)$.
\end{redrule}
\begin{redrule} \label{redrule:copies-comp}
    Let $(G, w)$ be an instance of MWBS and $v \in {\cal G}(G)$ of degree $\geq 2$. Let $(G', w)$ be the new instance, where in $G'$ we replace each edge $e = (u, v)$ (resp.~$e = (v, u)$) where $u \in {\cal G}(G)$ with another edge $e' = (u, x_{uv})$ (resp.~$e' = (x_{uv}, u)$), where $x_{uv}$'s are distinct vertices created for each such edge, and each $e'$ is embedded within the embedding of $e$, where $w(e')=w(e)$ (see \Cref{fi:red3}). 
\end{redrule}
\begin{figure}[tb]
    \centering
    \subfigure[]{\includegraphics[page=1,width=.25\textwidth]{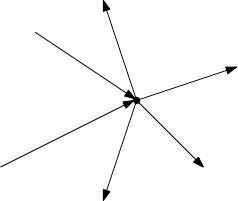}\label{fi:red3-a}}
    \hfil
    \subfigure[]{\includegraphics[page=2,width=.25\textwidth]{red3}\label{fi:red3-b}}
    \caption{A bimodal vertex (a) before and (b) after Reduction rule \ref{redrule:copies-comp} is applied.}
    \label{fi:red3}
\end{figure}

\begin{restatable}[*]{cl}{cleasysoundnesscomp} \label{cl:easy-soundness-comp}
    Reduction rules \ref{redrule:isolated-comp}, \ref{redrule:bimodal-comp} and \ref{redrule:copies-comp} are sound.
\end{restatable}

By applying Reductions \ref{redrule:isolated-comp}, \ref{redrule:bimodal-comp} and \ref{redrule:copies-comp} exhaustively, we get \cref{lem:reduction-to-simple-instance}, which is already enough to give a subexponential FPT algorithm by $b$ (\cref{thm:mwbs-fpt-b}). 

\begin{restatable}[*]{lemma}{lemreductiontosimpleinstance}\label{lem:reduction-to-simple-instance}
	Given an instance $(G, w)$ of MWBS, there exists a polynomial-time algorithm to obtain an equivalent instance $(G', w)$ with $G'$ being a subgraph of $G$, such that
	(i) $|{\cal B}(G')| \leq |{\cal B}(G)|$, 
	(ii) ${\cal G}(G')$ is an independent set in $G'$,
	and 
	(iii) for all $v \in {\cal G}(G')$, $\deg(v) = 1$ in the underlying graph of $G'$.
\end{restatable}


\begin{theorem} \label{thm:mwbs-fpt-b}
	There exists an algorithm that solves MWBS$(G, w)$ with $b$ bad vertices in $2^{\Oh(\sqrt{b})} \cdot n^{\Oh(1)}$ time. 
\end{theorem}
\begin{proof}
 By \Cref{lem:reduction-to-simple-instance}, $(G,w)$ is equivalent to $(G',w)$ with at most $b$ vertices of degree $> 1$, which we can compute in polynomial time.
    This implies $\bw(G') = \Oh(\tw(G')) = \Oh(\sqrt{b})$, and we can apply \Cref{thm:mwbs-fpt-bw} to obtain an algorithm that computes a solution for $(G', w)$ in $2^{\Oh(\bw(G'))} |V(G')|^{\Oh(1)}$ time. 
\end{proof}

We now describe how we can partition, for a given input, all good-edge sections into edge sets in such a way that there exists an optimal solution in which every set is either contained or deleted completely, and the total number of sets is bounded in a function of $b$. We will then show how we can replace the sets with edge sets of size at most two. The main difficulty will be to ensure that sets that exclude each other continue to do so in the reduced instance.

\begin{restatable}[*]{lemma}{lemtransformsectionstocutsets}\label{lem:transform-sections-to-cutsets}
	Let $(G, w)$ be an instance of MWBS with $n$ vertices and $b$ bad vertices, such that ${\cal G}(G)$ is an independent set in $G$ and $\mathrm{deg}(v) = 1$ for all $v \in {\cal G}(G)$.
	Let further $v \in {\cal B}(G)$, and let $S$ be a good edge-section of $v$.
	Then $S$ can be partitioned into at most 26 sets $S_1, \dots, S_{26}$, such that for every optimal solution $G' \subseteq G$ of $MWBS(G,w)$, there exists an optimal solution $G^* \subseteq G$ of $MWBS(G,w)$, such that $G'$ and $G^*$ coincides on $G \setminus S$, and for every $i$, $S_i$ is either contained or removed completely in $G^*$.

    Further, there exists a partition $P_1, \dots, P_j$ of $\{S_1, \dots, S_{26}\}$, such that for all $P_i$:
    (1) $|P_i|\leq 2$, (2) the edges in $P_i$ are consecutive in $S$ and (3) if $P_i = \{\mathcal{S}_1, \mathcal{S}_2\}$, then $\mathcal{S}_1$ consists of outgoing edges of $v$ iff $\mathcal{S}_1$ consists of incoming edges of $v$, and at least one of $\mathcal{S}_1, \mathcal{S}_2$ does not form a set of consecutive edges in $S$.
\end{restatable}

To show this, we enclose $S$ in a curve $\phi$, and then compute for every given configuration $X$ the maximal subgraph $G'$ such that $v$ has configuration~$X$ in~$\phi$. This yields a set of at most 12 possible locations for switches between incoming and outgoing edges in $S$, which gives a partition of $S$ into at most 13 sets (corresponding to $P_1, \dots, P_j$) that do not contain a switch, and thus at most 26 sets that will not be separated by an optimal solution, corresponding to $S_1, \dots, S_{26}$.
We now describe a parameter-preserving reduction from MWBS~to~Cut-MWBS.

\begin{restatable}[*]{lemma}{lemtransformMWBStoCutMWBS}\label{lem:transform-MWBS-to-CutMWBS}
	Given an instance $(G,w)$ of MWBS with $b$ bad vertices, we can find in polynomial time an instance $(G', w, \mathcal{E})$ of Cut-MWBS, so that:
	(i) For every $\mathcal{E}_i \subseteq \mathcal{E}$ with $|\mathcal{E}_i| \geq 2$, there exists a bad vertex $v \in G'$ and a good edge-section $S$ of $v$, so that $\mathcal{E}_i$ is a subset of $S$ and $\mathcal{E}_i$ contains only outgoing or only incoming edges of $v$.
	(ii) $|{\cal B}(G')| \leq b$,
	(iii) $|\mathcal{E}| = \Oh(b^2)$,
	(iv) $(G,w)$ and $(G', w, \mathcal{E})$ have the same optimal cost,
    (v) there exists a partition $P_1, \dots, P_j$ of $\mathcal{E}$, such that $|P_i|\leq 2$ for all $P_i$,
    (vi) if $|P_i| = 1$, then the edges-set contained in $P_i$ is either an edge between two bad vertices, or there exists a bad vertex $v \in G'$ and good edge-section $S$ of $v$, such that the edges contained in $P_i$ are all consecutive in $S$,
    and, 
    (vii) if $|P_i| = 2$ with $P_i = \{\mathcal{E}_1, \mathcal{E}_2\}$, there exists some $v \in {\cal B}(G')$ and a good edge-section $S$ of $v$, such that the edges 
    in $P_i$ are all consecutive in $S$; and $\mathcal{E}_1$ consists of outgoing edges of $v$ if and only if $\mathcal{E}_1$ consists of incoming edges of $v$, and at least one of $\mathcal{E}_1, \mathcal{E}_2$ does not form a set of consecutive edges in $S$.
\end{restatable}
\begin{figure}[tb]
    \centering
    \subfigure[]{\includegraphics[page=1,width=.43\textwidth]{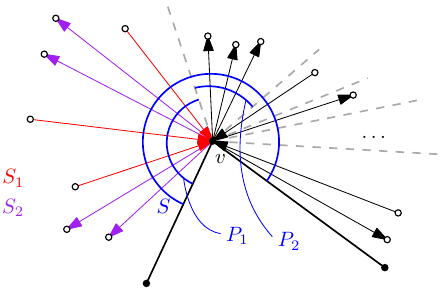}\label{fi:lem45-a}}
    \hfil
    \subfigure[]{\includegraphics[page=2,width=.43\textwidth]{lem4u5.pdf}\label{fi:lem45-b}}
    \caption{(a) Illustration for Lemmas  \ref{lem:transform-sections-to-cutsets} and \ref{lem:transform-MWBS-to-CutMWBS}. The gray dashed lines correspond to a set of switches between the optimal solution we will choose; they impose the partition $P_1, \dots, P_{13}$. $S_1$ ($S_2$) are the incoming (outgoing) edges of $P_1$, respectively. 
    (b) The same vertex after transition to \textsc{Cut-MWBS} by \Cref{lem:transform-MWBS-to-CutMWBS}, and after Reduction Rule \ref{redrule:melt-adjacent-edges} (\ref{redrule:simplify-cutsets}) got applied to $P_2$ ($P_1$), respectively.}
    \label{fi:lem45}
\end{figure}
See \Cref{fi:lem45-a} for a visualization. We obtain this transformation by applying \Cref{lem:reduction-to-simple-instance} in order to get a simplified equivalent instance $G'$. Let $E_\mathrm{rest}$ be all edges incident to two bad vertices. For every bad vertex $v$ and every good edges section $S$ of $v$, let $\mathcal{S}_{v, S}$ be the partition of $S$ obtained from \Cref{lem:transform-sections-to-cutsets}. We define $\mathcal{E} = \{e \mid e \in E_\mathrm{rest}\} \cup \bigcup_{v, S} \mathcal{S}_{v, S}$.
This defines the instance $(G', w, \mathcal{E})$ of Cut-MWBS. We will now further reduce the size of $(G,w,\mathcal{E})$.
\begin{redrule} \label{redrule:melt-adjacent-edges}
	Let $(G,w,\mathcal{E})$ be an instance of Cut-MWBS with  properties (i) to (vii) of \Cref{lem:transform-MWBS-to-CutMWBS}.
	Let  $v \in {\cal B}(G)$, let $S$ be a good edge-section of $v$, and let $\mathcal{E}_i \in \mathcal{E}$ such that $\mathcal{E}_i \subseteq S$ is a \emph{consecutive} set of edges in $S$.
	Then let $(G',w',\mathcal{E'})$ be the new instance that is obtained from $(G,w,\mathcal{E})$ by deleting all edges (and their incident good vertices) but one edge $e$ out of $\mathcal{E}_i$, and assigning $w'(e) = w(\mathcal{E}_i)$.
\end{redrule}

\begin{redrule} \label{redrule:simplify-cutsets}
	Let $(G,w,\mathcal{E})$ be an instance of Cut-MWBS with the properties (i) to (vii) of \Cref{lem:transform-MWBS-to-CutMWBS}.
	Let further $v\in {\cal B}(G)$, let $S$ be a good edge-section of $v$, and let $\mathcal{E}_\mathrm{in}, \mathcal{E}_\mathrm{out} \in \mathcal{E}$ such that $\mathcal{E}_\mathrm{in}, \mathcal{E}_\mathrm{out} \subseteq S$, $\mathcal{E}_\mathrm{in}$ are all incoming to $v$, $\mathcal{E}_\mathrm{out}$ are all outgoing of $v$, $\mathcal{E}_\mathrm{in} \cup \mathcal{E}_\mathrm{out}$ is a consecutive set of edges in $S$, and at least one of $\mathcal{E}_\mathrm{in}$ or $\mathcal{E}_\mathrm{out}$ does not form a consecutive set of edges in $S$.
	We construct a new edge-set $e_1, e_2, e_3, e_4$ as follows: $e_1, e_3$ are incoming for $v$, $e_2, e_4$ are outgoing of $v$, and all of $e_1, e_2, e_3, e_4$ are incident to a newly inserted (good) vertex $v_{e_k}$ for $k \in \{1,\dots, 4\}$. We set $w'(e_1) = w'(e_4) = 0$, $w'(e_2)= w(\mathcal{E}_\mathrm{out})$ and $w'(e_3)= w(\mathcal{E}_\mathrm{in})$. Further we assign $e_1, e_3 \in \mathcal{E}_\mathrm{in}$ and $e_2, e_4 \in \mathcal{E}_\mathrm{out}$.
	Let $(G',w',\mathcal{E})$ be the new instance that is obtained from $(G,w,\mathcal{E})$ by replacing the edges in $\mathcal{E}_i \cup \mathcal{E}_j$ with the consecutive sequence $e_1, e_2, e_3, e_4$.
\end{redrule}

\begin{restatable}[*]{cl}{claimsoundnesofcutsetreducings}\label{claim:soundnes-of-cutset-reducings}
	Reductions \ref{redrule:melt-adjacent-edges} and \ref{redrule:simplify-cutsets} are sound.
\end{restatable}

\begin{restatable}[*]{lemma}{lemboundnumberofedgesinCutMWBS}\label{lem:bound-number-of-edges-in-CutMWBS}
	Let $(G,w,\mathcal{E})$ be an instance of Cut-MWBS with $b$ bad vertices and properties (i) till (vii) of \Cref{lem:transform-MWBS-to-CutMWBS}.
	Then we can compute in polynomial time an equivalent instance $(G',w',\mathcal{E}')$ such that $V(G') = \Oh(b^2)$.
\end{restatable}
See \Cref{fi:lem45-b} for an illustration. We compute $(G',w',\mathcal{E}')$ by applying Reductions \ref{redrule:melt-adjacent-edges} and \ref{redrule:simplify-cutsets} exhaustively. 
To bound the size of the weights $w$, we use the approach of Etscheid et al.~\cite{EtscheidKMR17} and the well-known \cref{thm:franktardos}. 
This yields the compression of MWBS (\cref{thm:weightedcompression-b}) and a kernel for MWBS (\cref{thm:kernel-MWBS}).

\begin{theorem}[\cite{frank1987application}] \label{thm:franktardos}
    There is an algorithm that, given a vector $\omega \in \mathbb{Q}^{r}$ and an integer $N$, in polynomial time finds a vector $\bar{\omega}$ such that $||\bar{\omega}||_{\infty} = 2^{\Oh(r^3)}$ and $\text{sign}(\omega \cdot b) = \text{sign}(\bar{\omega} \cdot b)$ for all vectors $b \in \mathbb{Z}^r$ with $||b||_{1} \le N-1$.
\end{theorem}

\begin{restatable}[*]{theorem}{thmweightedcompressionb} \label{thm:weightedcompression-b}
    There exists a polynomial-time algorithm that, given an instance $(G,w)$ of MWBS with $b$ bad vertices and a target value $W$, computes an instance $(G',w',\mathcal{E})$ of Cut-MWBS with size $\Oh(b^8)$, and a new target value $W'$ with size $\Oh(b^6)$, such that there exists a solution for $(G,w)$ of cost $W$ if and only if there exists a solution for $(G',w',\mathcal{E})$ of cost $W'$.
\end{restatable}


\begin{restatable}[*]{theorem}{thmkernelMWBS}\label{thm:kernel-MWBS}
 The decision version of MWBS parameterized by the number of bad vertices $b$ admits a polynomial kernel.	
%
\end{restatable}
\section{Efficient PTAS for MWBS and Final Remarks}\label{sec:EPTAS-MWBS}
We sketch our Efficient Polynomial-Time Approximation Scheme (EPTAS) for MWBS, i.e., a $(1-\epsilon)$-approximation that runs in $2^{\Oh(1/\epsilon)} \cdot n^{\Oh(1)}$ time. 
%
We use Baker's technique \cite{Baker94} to design our EPTAS. Our goal is to reduce the problem to (multiple instances of) the problem, where the treewidth (hence, branchwidth) of the graph is bounded by $\Oh(1/\epsilon)$, at the expense of an $\epsilon$-factor loss in cost. Then, we can use our single-exponential algorithm in the branchwidth to solve each such instance exactly, which implies a $(1-\epsilon)$-approximation.

We sketch the details of this reduction. W.l.o.g. assume that the graph is connected. We perform a breadth-first search starting from an arbitrary vertex $v \in V(G)$, and partition the vertex-set into layers $L_0, L_1, \ldots$, where $L_i$ is the set of vertices at distance \emph{exactly} $i$ from $v$ in the \emph{undirected} version of $G$. It is known that the treewidth of the subgraph induced by any $d$ consecutive layers is upper bounded by $\Oh(d)$ -- this follows from a result of Bodlaender \cite{Bodlaender98}, which states that the treewidth of a planar graph with diameter $D$ is $\Oh(D)$. Let $t = 1/\epsilon$, and for each $0 \le i \le t$, let $E^{(i, i+1)}$ denote edges $uv$ such that $u \in L_j$, $v \in L_{j+1}$ with $j \mod t = i$. By an averaging argument, there exists an index $0 \le i \le t$, such that the total contribution of all the edges from an optimal solution (i.e., the set of edges inducing a maximum-weight bimodal subgraph) that belong to $E^{(i, i+1)}$, is at most $1/t = \epsilon$ times the weight of the optimal solution. Since we do not know this index $i$, we consider all values of $i$, and consider the subproblems obtained by deleting the edges. Then, the graph breaks down into multiple connected components, and the treewidth of each component is $\Oh(1/\epsilon)$. We solve each such subproblem optimally in time $2^{\Oh(1/\epsilon)} \cdot n^{\Oh(1)}$ using \Cref{thm:mwbs-fpt-bw}, and combine the solutions for the subproblems to obtain a solution for the original instance. Note that the graph obtained by combining the optimal solutions for the subproblems is bimodal, and for the correct value of $i$, the weight of the graph is at least $1-\epsilon$ times the optimal cost. That is, the combined solution is a $(1-\epsilon)$-approximation.

\begin{restatable}[*]{theorem}{eptas} \label{thm:mwbs-eptas}
    There exists an algorithm that runs in time $2^{\Oh(1/\epsilon)} \cdot n^{\Oh(1)}$ and returns a $(1-\epsilon)$-approximate solution for the given instance of MWBS. That is, MWBS admits an EPTAS.
\end{restatable}
We note that Baker's technique can also be used to obtain an EPTAS with the similar running for the \emph{minimization} variant of MWBS.
Although the high level idea is similar, the details 
are more cumbersome.

\medskip\noindent{\bf Final Remarks.}
%
We conclude by suggesting some open questions. One natural problem is to ask for a maximum $k$-modal subgraph for any given even integer $k \geq 2$; we believe that our ideas can be extended to this more general setting.
Another natural variant of MBS is to limit the number of edges that we can delete to get a bimodal subgraph by an integer $h$; in this setting, $h$ becomes another parameter in addition to those we have considered.
Finally, studying MBS in the variable embedding setting is an interesting future direction. 
%

\bibliographystyle{splncs04}
\bibliography{main}

\newpage
\appendix

\section{Appendix}

\subsection{Details for Section \ref{se:basic}}\label{se:app-basic}

We refer to the books~\cite{CyganFKLMPPS15,FominLSZ19} for an introduction to the parameterized complexity area. 
Formally, a \emph{parameterized problem} is a language $L\subseteq\Sigma^*\times \mathbb{N}$ where $\Sigma$ is a finite alphabet. Thus, an input of $L$ is a pair $(I,k)$
where $I$ is a string encoding the instance and $k\in\mathbb{N}$ is a \emph{parameter}. 
The computational complexity is measured as a function of $|I|$ and $k$. A problem $L$ is said to be \emph{fixed-parameter tractable} (FPT) if it can be solved in $f(k)\cdot |I|^{\Oh(1)}$ time for some function $f$. 

A \emph{kernelization} algorithm or \emph{kernel} for a parameterized problem $L$ is a polynomial-time algorithm that, given an instance $(I,k)$ of $L$, outputs an instance $(I',k')$ of the same problem such that (i) $(I,k)\in L$ if and only if $(I',k')\in L$ and (ii) $|I'|+k'\leq f(k)$ for a computable function $f$. The function $f$ is called the \emph{size} of the kernel; a kernel is polynomial of $f$ is a polynomial. Similarly, a \emph{compression} of $L$ into a (non-parameterized) problem $L'$ is a polynomial-time algorithm that for an instance $(I,k)$ of $L$, outputs an instance $I'$ of $L'$ such that (i) $(I,k)\in L$ if and only if $I'\in L$ and (ii) $|I'|\leq f(k)$ for a computable function $f$.


\subsection{Details for Section \ref{se:fpt-branchwidth}}\label{se:app-fpt-branchwidth}

\lemconfigsolutionbase*
\begin{proof}
    Since the inside of $\phi$ contains only $e$, testing whether non-deleting $e$ fulfills both $X_u$ and $X_v$ can be done in constant time. Knowing whether or not $e$ needs to be deleted gives us $G'$.
\end{proof}

\lemconfigsolutiondpstep*
\begin{proof}
    We will show that an optimal solution for $\phi_3$ with the given configuration set can be obtained from optimal solutions for $\phi_1$, $\phi_2$.
    First, we show that there exist configurations $X_{\phi_1}$ and $X_{\phi_2}$ such that the optimal subgraph $G'$ is optimal in $\phi_1$ and $\phi_2$ with regard to these configurations as well.
    
    Assume its not optimal in $\phi_1$.
    Since $G'$ is bimodal in $\phi_3$, we know that it has a configuration set $X_{\phi_1}$ for $\phi_1$, therefore, the restriction of $G'$ to the left side of $\phi_1$ is not an optimal solution under the condition that is bimodal in $\phi_1$ and has $X_{\phi_1}$ as configuration set.

    Now let $G_{\phi_1}$ be an optimal solution under these conditions, and let $G''$ be the subgraph, that coincides with $G_{\phi_1}$ on the inside of $\phi_1$ and with $G'$ on the outside.
    This solution has a lower cost then $G'$, since the cost on the inside of $\phi_1$ got lower, and the cost on the rest stays the same as in $G'$.

    To show that $G''$ is a valid solution, we need to show that it is bimodal in $\phi_3$ and has configuration $X_{\phi_3}$.
    To show that $G''$ is bimodal it is sufficient to consider those vertices that are cut by both $\phi_1, \phi_2$, since the bimodality of all other vertices inside of $\phi_3$ follows from the bimodality of $G'$ and $G_{\phi_1}$. 
    Let $v$ be a vertex cut by both $\phi_1$ and $\phi_2$. Let $X_v^1, X_v^2$ the configuration of $v$ regarding $\phi_1, \phi_2$ respectively. Since $v$ is bimodal in $G'$, $X_v^1$ and $X_v^2$ are compatible. Since compatible configurations imply bimodality and $v$ still has those configurations in $G''$, it is bimodal in $G''$.
    In order to show that $G''$ is still of configuration $X_{\phi_3}$, it suffices to consider the vertices $v$ that are cut by all of $\phi_1, \phi_2, \phi_3$, since the rest of the vertices on $\phi_3$ has the same configuration as in $G'$ and thus as in $X_{\phi_3}$. Let $v$ be such a vertex, let $X_v^1, X_v^2$ be the minimal of $v$ regarding $\phi_1, \phi_2$ respectively, and let $X_v$ be the configuration of $v$ regarding $\phi_3$. Since $G'$ is a valid solution, $X_v^1$ and  $X_v^2$ are compatible with respect to $X_v$. Since we did not change the configuration of $v$ in $G''$ regarding $\phi_1$, and the ability to form a configuration depends only on that, $v$ is of configuration $X_v$ for $\phi_3$ in $G''$ as well.
    This is a contradiction to the assumption that $G'$ was optimal.

    The same argument also shows that, given the right configuration sequence $X_{\phi_1}$ for $\phi_1$ and any arbitrary optimal solution $G_{\phi_1}$, there exists an optimal solution $G^*$ which is bimodal inside of $\phi_3$, has the required configuration sequence $X_{\phi_3}$ and coincides with $G_{\phi_1}$ on the inside of $\phi_3$.

    Since the situation is the symmetrical for $\phi_2$, we can find the optimal solution for the given configuration set $X_{\phi_3}$ by exhaustively trying all combinations of configuration sets for $\phi_1, \phi_2$ and computing the optimal value obtainable this way. Since $\phi_1, \phi_2$ have width at most $\ell$, and there are at most 6 configurations per vertex, the number of configurations is bounded by $6^{2\ell}$.
    Since we already know optimal solutions for all configuration sets of $\phi_1$, $\phi_2$, we can compute the optimal solution of every valid configuration set by computing the cost for $\phi_1$ and $\phi_2$, and taking the solution with the lowest overall cost, which requires only polynomial time. Hence, the overall running time is $\Oh(6^{2\ell})\cdot n^{\Oh(1)}$.
\end{proof}

\thmmwbsfptbw*
\begin{proof}
    Let $(G,w)$ be an instance of MWBS with $n$ vertices. Without loss of generality, $G$ is embedded in the sphere $\Sigma$. Let $\langle T, \xi, \Pi \rangle$ be a sphere-cut decomposition of $G$ with minimum width $\ell$, which can be computed in $\Oh(n^3)$ time.

    Without loss of generality, $G$ is connected (otherwise, the optimal solution can be computed for each connected component independently).
    If $G$ has branchwidth 1, then it is a disjoint union of stars~\cite{RobertsonS91}. For a star, we can compute the optimal bimodal subgraph in quadratic time by trying all combinations to switch from incoming to outgoing edge for $v$, all other vertices are already bimodal since they have only one incident edge.
    We assume therefore that $G$ has branchwidth $\ell > 1$, and know from \Cref{th:sphere-cut-with-deg-1} that we can obtain an optimal sphere-cut decomposition $\langle T, \xi, \Pi \rangle$ of $G$ in $\Oh(n^3)$ time. Let $\ell$ be the width of $\langle T, \xi, \Pi \rangle$.
    
    Let $r \in T$ be an arbitrary leaf of $T$ that we choose as a root, and let $e_r$ be the edge of $G$ associated with $r$. In this way we can assume that $T$ is rooted and directed, such that $r$ has outdegree 0 and every other node in $T$ has outdegree precisely 1. 
    We know that every arc $a \in E(T)$ corresponds to a noose $O$ that cuts precisely $\midset(a)$ in order $\pi_a$.
    Every noose $O$ bounds two closed discs in $\Sigma$, we define the inside of every $O$ to be the one that does not contain $e_r$.
    We now describe a dynamic program that computes an optimal solution for $(G,w)$.
    Let $O$ be a noose and let $X_O$ be a configuration set for $O$. Then we define $E_{(O, X_O)}$ to be an edge set of minimum weight, such that $G \setminus E_{(O, X_O)}$ is bimodal inside of $O$ and has the configuration set $X_O$ regarding $O$.
    We will now show how the entries $E_{(O, X_O)}$ can be computed bottom-up for every noose in the sphere-cut decomposition.
    
    We will have to start with the entries of leaves in $T$.
    Let $u \in V(T), u\neq r$ be a leaf of $T$, let $e = (v_1, v_2)$ be the edge in $G$ associated with $u$, let $a$ be the unique outgoing arc of $u$, and let $O_a$ be the noose corresponding to $a$.
    We can see that $O_a$ is a closed curve with $\midset(a) = \{v_1, v_2\}$, and the inside of $O_a$ contains precisely $e$.
    Now let $X_{O_a}$ be a configuration set for $O_a$. According to \Cref{lem:config_solution_base} we can compute $E_{(O_a, X_{O_a})}$ in constant time.
    Since $|\midset(a)| = 2$, we have precisely $6^2 = 36$ and thus a constant number of configuration sets for $O_a$, and can compute all entries for $O_a$ in $\Oh(1)$ time.

    Now let $u \in V(T)$ be an internal vertex of $T$, let $a$ be the unique outgoing arc of $u$, let $O_a$ be the noose corresponding to $a$. Let $E_{(O', X_{O'})}$ already be computed for every noose $O'$ corresponding to an arc in the sub-tree rooted by $u$ and for every configuration set $X_{O'}$ of $O'$.
    We know that $O_a$ is a closed curve with width at most $\ell$. We further know that the inside of $O_a$ is partitioned by two nooses that also have width at most $\ell$.
    Now let $X_{O_a}$ be a configuration set for $O_a$. According to \Cref{lem:config_solution_dpstep} we can compute $E_{(O_a, X_{O_a})}$ in $\Oh(6^{2\ell}) \cdot n^{\Oh(1)}$ time.
    Since $|\midset(a)| \leq \ell$, we have at most $6^{\ell}$ configuration sets for $O_a$, and can thus compute all entries for $O_a$ in $\Oh(6^{2\ell})\cdot n^{\Oh(1)}$ time.
    Since every arc of $T$ is either incoming for a leaf or an internal vertex of $T$, all table entries can be computed in this manner. Since $|E(T)| = \Oh(|E(G)|) = \Oh(|V(G)|)$, this can be done in $\Oh(6^{2\ell})\cdot n^{\Oh(1)}$ time.

    Recall that $e_r = (v_o, v_i)$ is the edge associated with the root $r$ of $T$. Let $a$ be the incident (incoming) arc of $r$ in $T$, and let $O_a$ be the noose associated with $a$.
    Assume that there exists an optimal solution $G'$ for MWBS$(G,w)$ such that $e_r \notin E(G')$, and let $E_\mathrm{opt} = E(G) \setminus E(G')$
    Then $G'$ is bimodal in $O_a$.
    On the other hand, every subgraph of $G$ that is bimodal inside of $O_a$ is an optimal solution for $(G,w)$. 
    Let $E^*_{\setminus e_r} = \min_{X_{O_a}}\{E_{(O_a, X_{O_a})}\}$, with $X_{O_a}$ being a configuration set of $O_a$.
    Since every subgraph of $G$ that is bimodal inside of $O_a$ has a configuration set for $O_a$, the graph $G^*_{\setminus e_r} = (V, E \setminus E^*_{\setminus e_r})$ is an optimal solution for MWBS$(G,w)$.
    Since all table-entries are known, and there are at most $36$ configuration sets for $O_a$, we can compute $E^*_{\setminus e_r}$ and thus $G^*_{\setminus e_r}$ in linear time.

    Now assume that there exists an optimal solution $G'$ for $(G,w)$ with $e_r \in E(G')$, and let $E_\mathrm{opt} = E(G) \setminus E(G')$.
    Then $G'$ is bimodal in $O_a$, has configuration $(i, o, i)$ in $v_i$, and has configuration $(o, i, o)$ in $v_o$.
    On the other hand, every subgraph of $G$ bimodal inside of $O_a$, containing $e_r$, with configuration $(i, o, i)$ in $v_i$ and configuration $(o, i, o)$ in $v_o$ regarding $O_a$ is an optimal solution for MWBS$(G,w)$.
    Set the configuration set $X_{O_a} = (X_{v_i} = (i, o, i), X_{v_o} = (o, i, o))$. We then know that $E^*_{e_r} = E_{(O_a, X_{O_a})}$ is an edge set of minimal cost such that $G^*_{e_r} = (V(G), E(G)\setminus E^*_{\setminus e_r})$ is an optimal solution for MWBS$(G,w)$.
    Since all table-entries are known, we already have $E^*_{\setminus e_r}$ and can thus compute $G^*_{\setminus e_r}$ in linear time.

    Since one of these two cases must be fulfilled in any optimal solution, one of $G^*_{\setminus e_r}$ and $G^*_{e_r}$ must be an optimal solution for MWBS$(G,w)$, and we can find it in linear time.
    
    Since we could compute a sphere-cut decomposition $\langle T, \xi, \Pi \rangle$ in $\Oh(n^3)$ time, since all table-entries $E_{(O, X_O)}$ could be computed in $\Oh(6^{2\ell})\cdot n^{\Oh(1)}$ time, and since $G^*_{e_r}$ and $G^*_{e_r}$ could be computed in linear time, the overall running time of the algorithm is $2^{\Oh(\ell)} \cdot n^{\Oh(1)}$.
\end{proof}

\subsection{Details for Section \ref{se:fpt-b}}\label{se:app-fpt-b}

\cleasysoundnesscomp*
\begin{claimproof}
The soundness of Reduction rule \ref{redrule:isolated-comp} is straightforward. 

Now consider Reduction rule \ref{redrule:bimodal-comp}. Consider an edge $(u, v)$ between two bimodal vertices $u$ and $v$. Suppose that  $G^*$ a bimodal subgraph of $G$ of weight at least $W$. Then $\hat{G}=G^*-(u,v)$ is a bimodal subgraph of $G'$ and $w(E(\hat{G}))\geq W-w(u,v)=W'$. 
In the reverse direction, let $\hat{G}$ be a bimodal subgraph of $G'$ of weight at least $W'$. Without loss of generality we can assume that $V(\hat{G})=V(G)$. Consider $G^*$ obtained from $\hat{G}$ by the addition of $(u,v)$. We have that  $G^*$ is bimodal because $u$ and $v$ are bimodal in $G$. Thus, $G^*$ is a bimodal subgraph of $G$ of weight at least $W=W'+w(u,v)$.

Finally, we show that Reduction rule~\ref{redrule:copies-comp} is sound. For this, we assume that $v\in\mathcal{G}(G)$ and $G'$ is obtained from $G$ by the replacement of $(u,v)\in E(G)$ by $(u,x_{xy})$ for all $u \in V$ that are adjacent to $v$, where $x_{uv}$ is the vertex constructed for $(u,v)$. 

Let $G^*$ be a bimodal subgraph of $G$ of maximum weight, and let $u\in V$, $u$ adjacent to $v$. If $(u,v)\notin E(G^*)$, we define $\hat{G}$ to be the spanning subgraph of $G'$ with $E(\hat{G})=E(G^*)$ and observe that $\hat{G}$ is a bimodal subgraph of $G$ because $x_{uv}$ is an isolated vertex of $\hat{G}$. Trivially, $w(E(\hat{G}))\geq w(E(G^*))$. 
Otherwise, if $(u,v)\in E(G^*)$, we construct the spanning subgraph $\hat{G}$ of $G'$ by setting 
$E(\hat{G})=(E(G^*)\setminus\{(u,v)\})\cup\{(u,x_{uv})\}$. 
Notice that $w(E(\hat{G}))\geq w(E(G))$. Furthermore, $x_{uv}$ is bimodal because it is incident only to $(u,x_{uv})$. This implies that $\hat{G}$ is a bimodal subgraph of $G'$ whose weight is at least the weight of $G$. 
 
For the opposite direction, assume that $\hat{G}$ is a bimodal subgraph 
of $G'$ of maximum weight. If $(u,x_{uv})\notin E(\hat{G})$, we set $G^*$ be the spanning subgraph of $G$ with $E(G^*)=E(\hat{G})$. Then $G^*$ is bimodal and $w(E(G^*))\geq w(E(\hat{G}))$. Suppose that $(u,x_{uv})\in E(\hat{G})$. Then we define $G^*$ to be the spanning subgraph of $G$ with $E(G^*)=(E(\hat{G})\setminus\{(u,x_{uv})\})\cup \{(u,v)\}$.
By definition, $w(E(G^*))=w(E(G'))$. Because $v$ is a bimodal vertex of $G$, we have that $v$ is bimodal in $G^*$. Therefore, $G^*$ is a bimodal subgraph of $G$ whose weight is at least the weight of $\hat{G}$. Thus, Reduction rule~\ref{redrule:copies-comp} is sound.
This concludes the proof.
\end{claimproof}


\lemreductiontosimpleinstance*
\begin{proof}
	Let $(G', w)$ be the instance of MWBS obtained by applying reductions \ref{redrule:isolated-comp}, \ref{redrule:bimodal-comp} and \ref{redrule:copies-comp} exhaustively.
	Since none of them introduces new bad vertices, (i) is fulfilled.
	Let $v\in {\cal G}(G')$. If $v$ were adjacent to another good vertex, we could apply reduction \ref{redrule:bimodal-comp}. Thus, (ii) is fulfilled.
	If $deg(v) > 1$, reduction \ref{redrule:copies-comp} could be applied. If $deg(v) = 0$, reduction \ref{redrule:isolated-comp} could be applied. Thus, (iii) is fulfilled.	
    It is left to show that the number of reduction steps is bounded in a polynomial. Reduction \ref{redrule:isolated-comp} is performed at most once for every already existing vertex and at most once for every vertex introduced by an iteration of reduction \ref{redrule:copies-comp}. Since the number of new vertices introduced over all iterations of reduction \ref{redrule:copies-comp} is bounded in $2\cdot |E(G)|$, reduction \ref{redrule:isolated-comp} is performed at most $|V(G)| + 2|E(G)|$ times.
    Reduction \ref{redrule:bimodal-comp} is performed at most $|E|$ times, since it always deletes an edge, and the number of edges stays the same during the other reductions.
    Reduction \ref{redrule:copies-comp} is performed at most once per preexisting vertex and never for any new vertex introduced during an iteration of it, so it is performed at most $|V(G)|$ times. 
    Thus, the total number of reductions steps is linear in the size of $G$, and the algorithm runs in polynomial time.
\end{proof}

\lemtransformsectionstocutsets*
\begin{proof}
	Since $S$ is a good edge-section, and all good incident vertices of $S$ have degree 1, we can draw a curve $\phi$ into $G'$ such that $\phi$ cuts $G$ only in $v$, and has precisely $S$ and the respective incident good vertices of $S$ in its interior.
	
	\begin{cl}\label{claim:config-for-v-solvable}
		Given a configuration $X$, we can in polynomial time compute a subset $S_X$ of $S$ of minimal cost, such that $v$ has configuration $X$ in $\phi$ if $S_X$ is removed.
	\end{cl}
	\begin{claimproof}
		Edge-deletion till fulfillment of a configuration is equivalent to finding the (at most two) switches from a group of incoming edges to a group of outgoing edges and vice versa.
		Since there are at most two switches, the number of possibilities is bounded by $\Oh(n^2)$.
		Once the placement of the switches is known, the corresponding costs are equivalent to the weight of the set of edges that are in the section of opposite type (incoming instead of outgoing or vice versa).
		$S_X$ is chosen as the set of edges with minimal cost obtained this way.
	\end{claimproof}

	Now let $G'$ be an optimal solution for $(G,w)$. We know that $v$ has a configuration in $G'$ regarding $\phi$, let $X$ be the minimal configuration for which this is true. Since bimodality of $v$ in $G'$ does not change if we change the inside of $\phi$ as long as $v$ keeps configuration $X$, we can construct another solution $G^*$ for $(G,w)$ that coincides with $G'$ on the outside of $\phi$ and with $G\setminus S_X$ on the inside of $\phi$. Since $S_X$ had the minimal cost for configuration $X$, the cost of $G^*$ is not bigger than the cost of $G'$, thus $G^*$ is an optimal solution for $(G,w)$.
	
	There are at most 6 configurations possible for $v$ regarding $\phi$, and every configuration is associated with at most 2 switches, therefore we can separate $S$ in 12 places and thus partition it in into 13 sets $S^1, \dots, S^{13}$, such that there exists an optimal solution $G^*$ for which no switch from incoming to outgoing edges (or vice versa) happens inside of a $S^i$.
	Now let $S^i$ be such a set, and let $\mathrm{In}(S^i)$ and $\mathrm{Out}(S^i)$ denote the incoming and outgoing edges of $S^i$, respectively.
	Since no switch happens inside of $S^i$, $\mathrm{In}(S^i)$ is either contained or removed completely in $G^*$. The same is true for $\mathrm{Out}(S^i)$.
	This gives us a partition $\{\mathrm{In}(S^i), \mathrm{Out}(S^i) \mid 1\leq i \leq 13\}$ of $S$ with at most 26 sets that has the wanted properties.

    We get the partition $P_1, \dots, P_j$ of $\{\mathrm{In}(S^i), \mathrm{Out}(S^i) \mid 1\leq i \leq 13\}$, if we set $P_i = \{\mathrm{In}(S^i), \mathrm{Out}(S^i)\}$ and then remove all empty sets.
\end{proof}

\lemtransformMWBStoCutMWBS*
\begin{proof}
	According to \Cref{lem:reduction-to-simple-instance}, we can find in polynomial time an equivalent instance $(G', w)$, such that $|{\cal B}(G')| \leq b$ (which already impies (ii)); ${\cal G}(G')$ being an independent set in $G'$, and for all $v \in {\cal G}(G')$, $\deg(v) = 1$. We will not change the graph further, this already implies (ii)
	We are now able to apply \Cref{lem:transform-sections-to-cutsets} to $(G', w)$.
	For every $v \in {\cal B}(G')$ and every good edge-section $S$ of $v$, let $\mathcal{S}_{v, S}$ denote a partition of $S$ as described in \Cref{lem:transform-sections-to-cutsets}.
	
	Let $E_\mathrm{rest}$ denote the set of all edges that are not contained in some $\mathcal{S}_{v, S}$.
	We define $\mathcal{E} = \{\cup_{v\in {\cal B}(G'), S} \mathcal{S}_{v, S}\} \cup \{\{e\}\mid e\in E_\mathrm{rest}\}$.
	Since every edge can be contained in at most one $\mathcal{S}_{v, S}$, $\mathcal{E}$ is a partition of $E(G')$.
	Thus, $(G',w,\mathcal{E})$ is an instance of Cut-MWBS.
	
	We show that $(G',w,\mathcal{E})$ has the described properties.
	Clearly, every solution for $(G',w,\mathcal{E})$ is one for $(G',w)$ as well. The other way around, we can find an optimal solution for $(G',w)$ that for any given $E_i \in\mathcal{E}$ either removes all edges in $E_i$ or none by iteratively applying \Cref{lem:transform-sections-to-cutsets}. This is an optimal solution for $(G',w,\mathcal{E})$ as well, and (iv) is fulfilled.
	
	Let $\mathcal{E}_i \subseteq \mathcal{E}$ with $|\mathcal{E}_i| \geq 2$. Then $\mathcal{E}_i \in \mathcal{S}_{v, S}$ for some $v \in {\cal B}(G')$ and some good edge-section $S$ of $v$, with only incoming or only outgoing edges of $v$ by construction, and thus (i) is fulfilled.
	
	Now let $e$ be an edge incident to a bad vertex $v$ and to some good vertex $u$. By definition, $e$ is contained in some good edge-section.
	\Cref{obs:bounded-sections} gives us that there exist at most $b \cdot (b-1)$ good edge-sections in total. Since every good edge-section got partitioned into at most 26 sets, we know that $|\{\cup_{v\in {\cal B}(G'), S} \mathcal{S}_{v, S}\}| = \Oh(b^2)$.
	Since the number of edges incident to two bad vertices is bounded by $\Oh(b)$, $|\mathcal{E}| = \Oh(b^2)$ as well and (iii) is fulfilled.
	
	Since we obtain every set $\mathcal{E}_i$ with more than 1 vertex out of some set $S^i$ of consecutive edges, and we take all edges of a specific type, (v) till (vii) are fulfilled by the properties of the partition  $P_1, \dots, P_j$ of $\{S_1, \dots, S_26\}$ of \Cref{lem:transform-sections-to-cutsets}.
\end{proof}

\claimsoundnesofcutsetreducings*
\begin{claimproof}
    If $\mathcal{E}_i \subseteq S$ is a consecutive set of edges in $S$, it clearly suffices to choose one representative for a consecutive set of edges $\mathcal{E}_i$ in $S$ that are all part of the same in-wedge or out-wedge, and reduction \ref{redrule:melt-adjacent-edges} is sound.
	
	To show soundness of \ref{redrule:simplify-cutsets}, consider an optimal deletion set $E_\mathrm{opt}$ for $(G,w,\mathcal{E})$, in which no switch of $v$ between incoming and outgoing vertices or vice versa happens in the consecutive edge sequence of $\mathcal{E}_\mathrm{in} \cup \mathcal{E}_\mathrm{out}$.
	Thus, $\mathcal{E}_\mathrm{in}$ is deleted in the optimal solution if and only if $\mathcal{E}_\mathrm{out}$ is not deleted in the optimal solution.
	Let without loss of generality $\mathcal{E}_\mathrm{in} \in E_\mathrm{opt}$. Thus, $E'_\mathrm{opt} = E_\mathrm{opt}\setminus \mathcal{E}_\mathrm{in} \cup \{e_1, e_3\}$ has the same cost as $E_\mathrm{opt}$, and it is a valid solution deletion set for $(G',w',\mathcal{E})$ since the outgoing vertices $e_2, e_4$ do not violate bimodality.
	We can obtain an optimal solution for $(G,w,\mathcal{E})$ out of one for $(G',w',\mathcal{E})$ in the same manner.
\end{claimproof}

\lemboundnumberofedgesinCutMWBS*
\begin{proof}
	Apply Reductions \ref{redrule:melt-adjacent-edges} and \ref{redrule:simplify-cutsets} exhaustively to $(G,w,\mathcal{E})$ in order to obtain $(G',w',\mathcal{E}')$.
	Then, $|\mathcal{E'}_i| \leq 2$ for all $\mathcal{E'}_i \in \mathcal{E'}$. Since $\mathcal{E'}$ is a partition of $E(G')$, we have $|E(G')| \leq 2|\mathcal{E'}| = |\mathcal{E}| = \Oh(b^2)$. Since none of the reduction rules creates isolated vertices, size of $V(G)$ is bounded by $\Oh(b^2)$ as well.
\end{proof}

\thmweightedcompressionb*
\begin{proof}
	We first apply \Cref{lem:transform-MWBS-to-CutMWBS} and then \Cref{lem:bound-number-of-edges-in-CutMWBS} to obtain an instance $(G,w,\mathcal{E})$ of Cut-MWBS with $|V(G)| = \Oh(b^2)$ and the same optimal cost as MWBS$(G,w)$.
    Let $r = |E(G)| + 1 = \Oh(b^2)$, and let $\omega = (w(e_1), w(e_2), \ldots, w(e_m), W) \in \mathbb{Q}^{r}$, where the edges in $G$ are indexed in an arbitrary order.
	Using \Cref{thm:franktardos}, we obtain a vector $\bar{\omega}$ such that $\text{sign}(\omega \cdot x) = \text{sign}(\bar{\omega} \cdot x)$ for each $x \in \mathbb{Z}^{r}$ with $||x||_1 \le r-1$.
	Now consider an edge set $F \subseteq E(G)$ upon whose deletion we obtain an optimal solution for $(G,w,\mathcal{E})$, and consider the vector $x_F \in \{-1, 0, 1\}^{r}$ such that 
 
    $$x_F(e_i) = \begin{cases}
    1 &\text{ if } e_i \in F,
    \\0 &\text{ if } e_i \not\in F
    \end{cases}$$
    and the last ($r$st) coordinate of $b_F$ is equal to $-1$. Then, note that $\omega \cdot x_F = \sum_{e_i \in F} w(e_i) - W$, which means that $F$ is a feasible solution if and only if $\omega \cdot x \le 0$, i.e., $\text{sign}(\omega \cdot x_F) \neq 1$. Since the new vector $\bar{\omega}$ preserves the sign of inner product with all such vectors $x_F$, it follows that the (in-)feasibility of all weighted solutions is preserved w.r.t.~the new weight vector $\bar{\omega}$. Thus, the new weights of the edges are given by the respective entries in $\bar{\omega}$, and the new target weight $\bar{W'}$ is given by the last entry of $\bar{\omega}$. Since each entry in $\bar{\omega}$ is an integer whose absolute value is bounded by $2^{\Oh(b^6)}$, it follows that we need $ \Oh(b^8)$ bits to encode $\bar{\omega}$.
\end{proof}

\thmkernelMWBS*
\begin{proof}
	We can compute an $(G',w',\mathcal{E})$ of Cut-MWBS with size $\Oh(b^8)$, and a new target value $W'$ with size $\Oh(b^6)$ in polynomial time, such that there exists an solution for $(G,w)$ of cost $W$ if and only if there exists a solution for $(G',w',\mathcal{E})$ of cost $W'$ according to \Cref{thm:weightedcompression-b}.
	Since the decision versions of MWBS and Cut-MWBS are both NP-complete, there exists a polynomial-time reduction from the decision version of Cut-MWBS onto the decision version of MWBS (see e.g. \cite[Theorem~1.6]{FominLSZ19}).
	Using this reduction, we can compute a new instance $(G^*, w^*)$ of MWBS and some $W^*$ such that there exists a solution for $(G^*, w^*)$ of cost at most $W^*$ if and only if there exists a solution for  $(G',w',\mathcal{E})$ of cost at most $W'$, such that the size of $(G^*, w^*)$ is bounded in a polynomial of the size of  $(G',w',\mathcal{E})$ and thus in a polynomial of $b$.
\end{proof}



\input{arx-appendix}

\end{document}

%% file: arx-appendix.tex
\section{EPTAS for MWBS} \label{sec:appendix-eptas}

In this section, we use Baker's technique \cite{Baker94} to prove the following theorem.
\eptas*
\begin{proof}

    For a directed graph $H$, let $\overline{H}$ denote its undirected version. Let $(G, w)$ be the given instance of MWBS. Let $E^* \subseteq E(G)$ denote an optimal solution, i.e., a maximum-weight subset of edges such that $G^* = G[E^*]$ is bimodal. Let $OPT = w(E^*)$ denote the weight of the solution. Without loss of generality, we assume $\gbar = (V, F)$ is connected -- otherwise we can use the following algorithm on each connected component separately. 

    Let $t = \lceil 1/\epsilon\rceil$. Fix an arbitrary vertex $v \in V(\gbar)$. For any $i \ge 0$, let $L_i$ denote the set of vertices that are at distance exactly $i$ from $v$ in $\gbar$. Note that any $u \in L_i$ has neighbors in $L_j$ such that $|i-j| \le 1$. For $i \ge 0$, let $E^{(i, i+1)} \subseteq E(G)$ denote the set of directed edges with one endpoint in $L_i$ and other endpoint in $L_{i+1}$. For each $E^{(i, i+1)}$, let $F^{(i, i+1)}$ denote the corresponding set of undirected edges.
    \begin{observation} \label{obs:averaging}
    For some $i \in \LR{0, 1, \ldots, t-1}$, it holds that $w(E^{(i, i+1)} \cap E^*) \le \frac{1}{t} \cdot OPT$.
    \end{observation}
    \begin{proof}
        Follows from the fact that the sets $E^{(i, i+1)}$ and $E^{(j, j+1)}$ are pairwise disjoint for distinct $i, j$. 
    \end{proof}

    For each $0 \le i \le t-1$, let $G_i$ denote the graph $G - E^{(i, i+1)}$. Note that each connected component in $\overline{G_i}$ is induced by vertices belonging to at most $t$ consecutive layers, and hence, using standard arguments (e.g., \cite{CyganFKLMPPS15}), the treewidth of $\overline{G_i}$ is bounded by $O(t) = O(1/\epsilon)$. Our algorithm performs the following operations for each $0 \le i \le t-1$. We solve MWBS on each connected component of $G_i$ separately, in time $2^{O(1/\epsilon)} \cdot n^{O(1)}$, using the algorithm from \Cref{thm:mwbs-fpt-bw}. Let $G'_i$ denote the graph obtained by taking the disjoint union of the solutions for all subgraphs. It is easy to see that the resulting graph is bimodal. Finally, for the correct value of $i$, i.e., that guaranteed by \Cref{obs:averaging}, we know that the total weight of edges in $|E^{(i, i+1)} \cap E^*| \le \epsilon \cdot OPT$; whereas the edges in $E^{*(i, i+1)} \setminus E^*$ do not contribute to $OPT$. It follows that the total weight of the edges in $G'_i$ is at least $(1-\epsilon) \cdot OPT$.
\end{proof}

Note that here we designed an approximation version for the maximization objective of MWBS, i.e., where we want to maximize the total weight of the bimodal subgraph. We can also consider the minimization objective, where the goal is to minimize the total weight of the edges removed to obtain a bimodal subgraph. Note that although the decision versions of the two objectives are equivalent, approximation for one version does not necessarily imply a good approximation for the other. Nevertheless, we can adapt Baker's technique to design an EPTAS, i.e., $(1+\epsilon)$-approximation for the minimization objective as well. However, a formal description of this algorithm is quite tedious, since simply deleting the edges of $E^{(i, i+1)}$ as in above, is not sufficient. Instead, we need to ``copy'' the edges of $E^{(i, i+1)}$ to subproblems on ``both sides'' in an appropriate manner. In the following theorem, we give a formal description of the EPTAS for the minimization variant.

\begin{theorem} \label{thm:mwbs-eptas-min}
    There exists an $(1+\epsilon)$-approximation for the minimization variant of MWBS that runs in time $2^{O(1/\epsilon)} \cdot n^{O(1)}$. That is, the minimization variant of MWBS admits an EPTAS.
\end{theorem}
\begin{proof}
    Let $t = \lceil 2/\epsilon\rceil$. Fix an arbitrary vertex $v \in V(\gbar)$. For any $i \ge 0$, let $L_i$ denote the set of vertices that are at distance exactly $i$ from $v$ in $\gbar$. Note that any $u \in L_i$ has neighbors in $L_j$ such that $|i-j| \le 1$. For $i \ge 0$, let $E^{(i, i+1)} \subseteq E(G)$ denote the set of directed edges that with one endpoint in $L_i$ and other endpoint in $L_{i+1}$. For each $E^{(i, i+1)}$, let $F^{(i, i+1)}$ denote the corresponding set of undirected edges.
    \begin{observation} \label{obs:averaging}
    For some $i \in \LR{0, 1, \ldots, t-1}$, it holds that $w(E^{(i, i+1)} \cap E^*) \le \frac{1}{t} \cdot OPT$.
    \end{observation}
    \begin{proof}
        Follows from the fact that the sets $E^{(i, i+1)}$ and $E^{(j, j+1)}$ are pairwise disjoint for distinct $i, j$. 
    \end{proof}

    \begin{figure}
        \centering
        \includegraphics[scale=0.8]{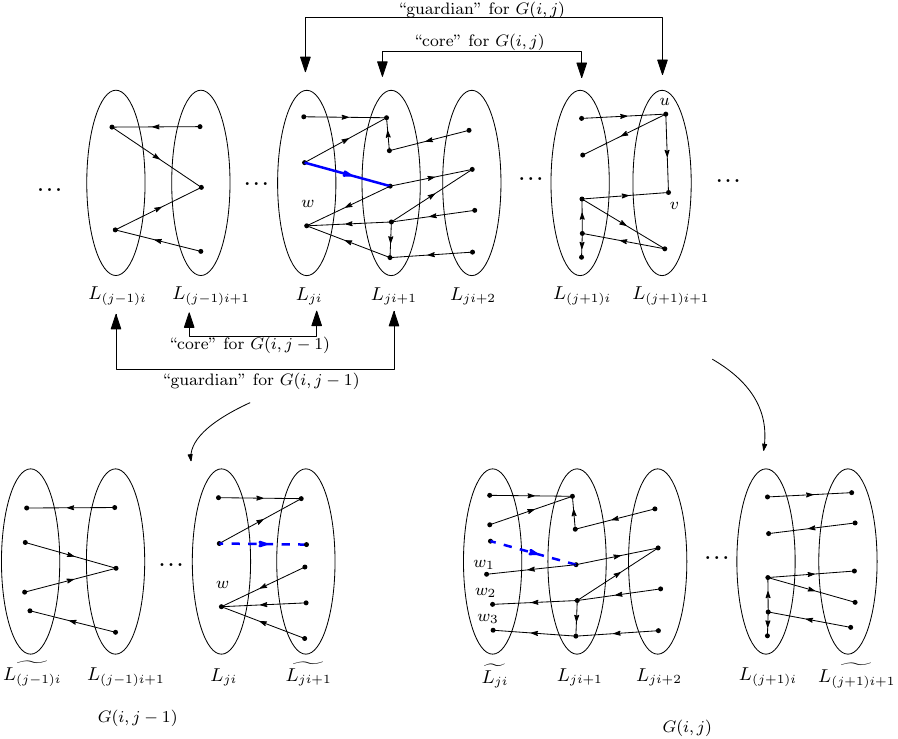}
        \caption{Construction of graphs $G(i, j-1)$ and $G(i, j)$ using layers. Observe that the vertex $w \in L_{ji}$ is split into multiple copies in $\widetilde{L_{ji}}$, and the edge $(u, v)$ where $u, v \in L_{(j+1)i+1}$ is not retained in $G(i, j)$ after making copies of the respective vertices. However, $w$ is kept as ``original'' in $L_{ji}$ in the graph $G(i, j)$. Finally, note that the blue edges in the original graph belongs to $E^{(i, i+1)}$, and corresponds to two copies -- one in $G(i, j-1)$ and the other in $G(i, j)$. This is why we incur the additive factor $2 \cdot 1/t \cdot OPT$ when combining the solutions.}
        \label{fig:my_label}
    \end{figure}

    For any integer $i \in \{0, 1, \ldots, t-1\}$, we construct a graph $G_i$ by taking a disjoint union of graphs $G(i, j)$ for $j \ge 0$, constructed as follows. 

    For $j = 0$, the graph $G(i, 0)$ is defined on vertex set $L_0 \cup L_1 \cup \ldots L_i$, along with a special vertex set $\tilde{L_{i+1}}$ defined later. All edges $e = (u, v)$, with $u, v \in L_0 \cup L_1 \cup \ldots \cup L_{i}$ are present in $G(i, 0)$. Next, consider edges with exactly one endpoint in $L_i$ and the other in $L_{i+1}$, and we ``split'' the vertices in $L_{i+1}$ to form $\tilde{L_{i+1}}$, such that each vertex in $\widetilde{L_{i+1}}$ has exactly one edge incident to it. More formally, consider an edge $e = (u, v)$ with $u \in L_i$ and $v \in L_{i+1}$. Then, we add a new vertex $w_{(u, v)}$ to $\widetilde{L_{i+1}}$, and we add an edge $(u, w_{(u, v)})$. Similarly, for an edge $(u', v')$ with $u' \in L_{i+1}$ and $v' \in L_i$, we add a new vertex $w_{(u', v')}$ to $\tilde{L_{i+1}}$, and add an edge $(w_{(u', v')}, v')$. Note that each vertex in $\tilde{L_{i+1}}$ has degree $1$, and is trivially bimodal. 

    For $j \ge 1$, the graph $G(i, j)$ is defined on the vertex set $\widetilde{L_{ji}} \cup L_{ji+1} \cup L_{ji+2} \ldots \cup L_{(j+1)i} \cup \widetilde{L_{(j+1)i+1}}$, where the sets $\tilde{L_{ji}}$ and $\tilde{L_{(j+1)i+1}}$ are obtained by ``splitting'' the vertices in $L_{ji}$ and $L_{(j+1)i+1}$, respectively. More formally, the edge set of $G(i, j)$ is constructed as follows. All edges with both endpoints in $L_{ji+1} \cup L_{ji+2} \ldots \cup L_{(j+1)i}$ are retained in $G(i, j)$. For each edge $(u, v)$ with $u \in L_{ji}$ and $v \in L_{ji+1}$, we add a new vertex $w_{(u, v)}$ in $\tilde{L_{ji}}$ and add an edge $(w_{(u, v)}, v)$. The other case with $(u', v')$ where $u' \in L_{ji+1}$ and $v' \in L_{ji}$ is defined analogously. Finally, the set $\tilde{L_{(j+1)i+1}}$ is defined analogously. Note that that newly created edge in $G(i, j)$ for $j \ge 0$ retains its original weight in $G_i$.  

    First, we observe that the treewidth of $G_i$ is bounded by $O(t)$, since it is obtained by taking at most $t+2$ consecutive BFS layers (\cite{CyganFKLMPPS15}). Furthermore, each edge $(u, v) \in E^{(i, i+1)}$ corresponds to exactly two copies in $G_i$ after the splitting process. Let $E'$ denote the set that includes both the copies of edges in $E^* \cap E^{(i, i+1)}$. Then, observe that $\tilde{E^*_i} \coloneqq (E^* \setminus E^{(i, i+1)}) \cup E'$ is a feasible solution for $G_i$, i.e., $G_i \setminus \tilde{E^*}$ is bimodal, and $w(\tilde{E^*_i}) \le w(E^*) + 2 \cdot w(E^* \cap E^{(i, i+1)})$. Furthermore, consider any feasible solution $F_i$ to $G_i$, and map it back to the original graph $G$ to obtain a solution $F$, as follows. If at least one copy of an edge is included in $F_i$, then we add it to $F$. We observe that $G \setminus F$ is bimodal, and $w(F) \le w(F_i)$. From this discussion and from \Cref{obs:averaging}, it follows that there exists some $i^* \in \{0, 1, \ldots, t-1\}$, such that $OPT(G_i) \le (1+2/t) \cdot OPT(G)$. 

    Now, our algorithm proceeds as follows. We create the graph $G_i$ for each $i \in \{0, 1, \ldots, t-1\}$, and use \Cref{thm:mwbs-fpt-bw} to find an optimal solution in time $2^{O(t)} \cdot n^{O(1)}$, and we map each optimal solution back to $G$ as described above. We output the minimum-weight solution found in this manner over all $i \in \{0, 1, \ldots, t-1\}$. From the previous paragraph, it follows the cost of this solution is at most $(1+2/t) \cdot OPT(G) \le (1+\epsilon) \cdot OPT(G)$.
\end{proof}